\documentclass[11pt]{article}

\newlength{\actualtopmargin}
\newlength{\actualsidemargin}
\usepackage[tbtags]{amsmath}
\usepackage{amsthm}
\usepackage{amssymb}
\usepackage{amsfonts}
\usepackage{graphicx}
\usepackage{hyperref}
\usepackage{url}

\theoremstyle{plain}
  \newtheorem{theorem}{Theorem}
  \newtheorem{lemma}[theorem]{Lemma}

\theoremstyle{definition}

\theoremstyle{remark}
  \newtheorem*{remark}{Remark}
\theoremstyle{plain}
  \newtheorem*{theorem*}{Theorem}
  \newtheorem*{lemma*}{Lemma}
  \newtheorem*{corollary*}{Corollary}
  \newtheorem*{proposition*}{Proposition}
  \newtheorem*{claim*}{Claim}

\newenvironment{step}
  {
    \begin{enumerate}

  }
  {\end{enumerate}}

\newenvironment{algorithm*}[1]
  {
    \begin{center}
      \hrulefill\\
      \textbf{#1}
  }
  {
    \vspace{-\baselineskip}
    \hrulefill
    \end{center}
  }

\newenvironment{protocol*}[1]
  {
    \begin{center}
      \hrulefill\\
      \textbf{#1}
  }
  {
    \vspace{-\baselineskip}
    \hrulefill
    \end{center}
  }

\newlength{\itemwidth}
\newlength{\descriptionwidth}
\newenvironment{promiseproblem*}[4]
  {
    \begin{center}
      \hrulefill\\
      \textbf{\textsc{#1}}
      \settowidth{\itemwidth}{\textbf{Yes Instances:}}
      \setlength{\descriptionwidth}{\textwidth}
      \addtolength{\descriptionwidth}{-\itemwidth}
      \addtolength{\descriptionwidth}{-\labelsep}
      \begin{description}
        \item[\parbox{\itemwidth}{Input:}]
          \parbox[t]{\descriptionwidth}{#2}
        \item[\parbox{\itemwidth}{Yes Instances:}]
          \parbox[t]{\descriptionwidth}{#3}
        \item[\parbox{\itemwidth}{No Instances:}]
          \parbox[t]{\descriptionwidth}{#4}
      \end{description}
  }
  {
    \vspace{-1\baselineskip}
    \hrulefill
    \end{center}
  }

\newcommand{\calC}{\mathcal{C}}

\newcommand{\classfont}{\mathsf}

\newcommand{\AM}{\classfont{AM}}
\newcommand{\BP}{\classfont{BP}}
\newcommand{\NP}{\classfont{NP}}
\newcommand{\NQP}{\classfont{NQP}}
\newcommand{\BPP}{\classfont{BPP}}
\newcommand{\BQP}{\classfont{BQP}}
\newcommand{\PH}{\classfont{PH}}
\newcommand{\PP}{\classfont{PP}}
\newcommand{\MA}{\classfont{MA}}

\newcommand{\AC}{\classfont{AC}}

\newcommand{\SBP}{\classfont{SBP}}
\newcommand{\SBQP}{\classfont{SBQP}}
\newcommand{\post}{\classfont{Post}}
\newcommand{\postBPP}{\classfont{PostBPP}}
\newcommand{\postBQP}{\classfont{PostBQP}}
\newcommand{\co}{\classfont{co}\textsf{-}}
\newcommand{\CequalP}{\classfont{C}_{=}\classfont{P}}
\newcommand{\coCequalP}{{\co\CequalP}}

\newcommand{\SBQfP}[1]{{\classfont{SBQ}_{[#1]}\classfont{P}}}
\newcommand{\NQfP}[1]{{\classfont{NQ}_{[#1]}\classfont{P}}}

\newcommand{\SBQoneP}{\SBQfP{1}}
\newcommand{\NQoneP}{\NQfP{1}}

\newcommand{\tensor}{\otimes}

\newcommand{\bra}[1]{\langle #1 \rvert}

\newcommand{\ket}[1]{\lvert #1 \rangle}

\newcommand{\ketbra}[1]{\lvert #1 \rangle \langle #1 \rvert}

\newcommand{\conjugate}[1]{{#1^\dagger}}

\newcommand{\abs}[1]{\lvert #1 \rvert}

\newcommand{\bigabs}[1]{\bigl\lvert #1 \bigr\rvert}

\newcommand{\norm}[1]{\lVert #1 \rVert}

\newcommand{\function}[3]{{#1 \colon #2 \to #3}}
\newcommand{\set}[1]{{\{ #1 \}}}
\newcommand{\Set}[2]{{\{ #1 \colon #2 \}}}

\newcommand{\init}{\mathrm{init}}
\newcommand{\final}{\mathrm{final}}
\newcommand{\acc}{\mathrm{acc}}

\newcommand{\DQCone}{\textrm{DQC1}}

\newcommand{\ignore}[1]{}

\begin{document}

\sloppy


\title{
  \Large{
    \textbf{
      Impossibility of Classically Simulating One-Clean-Qubit Computation
    }
  }
}

\author{
  Keisuke Fujii\footnotemark[1]\\
  \and
  Hirotada Kobayashi\footnotemark[2]\\
  \and
  Tomoyuki Morimae\footnotemark[3]\\
  \and
  Harumichi Nishimura\footnotemark[4]\\
  \and
  Shuhei Tamate\footnotemark[2]~\footnotemark[5]\\
  \and 
  Seiichiro Tani\footnotemark[6]
}

\date{}

\maketitle
\pagestyle{plain}
\setcounter{page}{1}

\renewcommand{\thefootnote}{\fnsymbol{footnote}}

\vspace{-5mm}

\begin{center}
\large{
  \footnotemark[1]The Hakubi Center for Advanced Research and Graduate School of Informatics\\
  Kyoto University, Kyoto, Japan\\
  [2.5mm]
  \footnotemark[2]%
  Principles of Informatics Research Division,  
  National Institute of Informatics, 
  Tokyo, Japan\\
  [2.5mm]
  \footnotemark[3]%
  ASRLD Unit, 
  Gunma University, 
  Kiryu, Gunma, Japan\\
  [2.5mm]
  \footnotemark[4]%
  Department of Computer Science and Mathematical Informatics\\
  Graduate School of Information Science, 
  Nagoya University, 
  Nagoya, Aichi, Japan\\
  [2.5mm]
  \footnotemark[5]%
  RIKEN Center for Emergent Matter Science, 
  Wako, Saitama, Japan\\
  [2.5mm]
  \footnotemark[6]%
  NTT Communication Science Laboratories, 
  NTT Corporation, 
  Atsugi, Kanagawa, Japan
}\\
[5mm]
\large{25 February 2015}\\
[8mm]
\end{center}

\renewcommand{\thefootnote}{\arabic{footnote}}


\begin{abstract}
Deterministic quantum computation with one quantum bit~(DQC1) 
is a restricted model of quantum computing
where the input state is the completely mixed state except for a single clean qubit,
and only a single output qubit is measured at the end of the computing.
It is proved that
the restriction of quantum computation to the DQC1 model
does not change the complexity classes~$\NQP$ and $\SBQP$. 
As a main consequence,
it follows that the DQC1 model cannot be efficiently simulated by classical computers 
unless the polynomial-time hierarchy collapses to the second level
(more precisely, to $\AM$),
which answers the long-standing open problem posed by Knill and Laflamme
under the very plausible complexity assumption.
The argument developed in this paper 
also weakens the complexity assumption
necessary for the existing impossibility results on classical simulation
of various sub-universal quantum computing models,
such as the IQP model and the Boson sampling.
\end{abstract}


\section{Introduction}
\label{Section: introduction}

\paragraph{Background.}

The \emph{deterministic quantum computation with one quantum bit~(DQC1)},
often mentioned as the \emph{one-clean-qubit model},
is a restricted model of quantum computing proposed by 
Knill and Laflamme~\cite{KniLaf98PRL} originally motivated 
by nuclear magnetic resonance (NMR) quantum information processing. 
A DQC1 computation over $n$~qubits starts with the initial state of
the completely mixed state except for a single clean qubit,
namely,
${\ketbra{0} \tensor \bigl(\frac{I}{2}\bigr)^{\tensor (n-1)}}$. 
After applying a polynomial-size quantum circuit
to this state,
only a single output qubit is measured in the computational basis
at the end of the computing in order to read out the computation result. 

The DQC1 model does not seem to be universal for quantum computation.
Indeed, it was shown to be non-universal 
under some reasonable assumptions~\cite{AmbSchVaz06JACM}.
Moreover, since any quantum computation on $\bigl(\frac{I}{2}\bigr)^{\otimes n}$
is trivial to simulate classically,
the DQC1 model looks easy to classically simulate at first glance. 
Surprisingly, however,
the DQC1 model turned out to be able to efficiently solve several problems
for which no efficient classical algorithms are known, 
such as calculations of an integrability tester~\cite{PouLafMilPaz03PRA}, 
the spectral density~\cite{KniLaf98PRL}, the fidelity decay~\cite{PouBluLafOll04PRL},
Jones and HOMFLY polynomials~\cite{ShoJor08QIC,JorWoc09QIC}, 
and an invariant of 3-manifolds~\cite{JorAla11TQC}. 
While the amount of entanglement is very limited, 
the DQC1 model does exhibit some non-classical correlations~\cite{DatFlaCav05PRA,DatShaCav08PRL,DatVid07PRA}.  
In short, the DQC1 model is believed to be
a very restricted, but still genuinely quantum, computing model
--- a \emph{sub-universal} quantum computing model whose power is 
something between classical computation and universal quantum computation.  
It is a long-standing open problem whether efficient classical simulation is possible for the DQC1 model,
posed already in the first paper of the DQC1 model by Knill and Laflamme~\cite{KniLaf98PRL}.

With the development of quantum algorithms, 
proving hardness of classical simulations of quantum computations,
even in restricted quantum computing models,
becomes a very fundamental methodology for clarifying the power of quantum computing. 
Recently, a number of studies
focused on the hardness of \emph{weak} simulation
of sub-universal quantum computing models under some reasonable assumptions~\cite{TerDiV04QIC,BreJozShe11RSPA,AarArk13ToC,NiVan13QIC,Bro14arXiv,JozVan14QIC,MorFujFit14PRL,TakYamTan14QIC,TakTanYamTan14arXiv}.
Namely, a plausible assumption in complexity theory 
leads to the impossibility of efficient sampling by a classical computer
according to an output probability distribution
generatable with a quantum computing model. 
Among them are the IQP model~\cite{BreJozShe11RSPA} and the Boson sampling~\cite{AarArk13ToC},
both of which are proved hard for classical computers to weakly simulate,
unless the polynomial-time hierarchy collapses to the third level.

An interesting question to ask is whether a similar result holds even for the DQC1 model. 
Very recently, Morimae, Fujii and Fitzsimons~\cite{MorFujFit14PRL} approached to answering the question. They focused on the DQC1$_m$ model,
the generalization of the DQC1 model
that allows $m$ output qubits to be measured at the end of the computation,
and showed that the DQC1$_m$ model with ${m \geq 3}$ cannot be weakly simulated
unless the polynomial-time hierarchy collapses to the third level.
Their proof is based on a postselection argument 
that fictitiously projects a state onto a specific branch of the superposition with unit probability.
The complexity class~$\postBQP$
corresponding to bounded-error quantum polynomial-time computations 
with postselection is known equivalent to $\PP$~\cite{Aar05RSPA}.  
A key technique in Ref.~\cite{MorFujFit14PRL} was
a DQC1$_m$-type computation with postselection that simulates any given quantum circuit: 
By using a generalized Toffoli gate and the postselection over a single clean qubit,
one can retrieve from the completely mixed state
the ideal state in which all the qubits are in state~$\ket{0}$. 
Hence, the output of any quantum circuit with postselection  
can be simulated by the DQC1$_3$-type computation by measuring three qubits in total:
one qubit is measured for the postselection to pick up the ideal all-zero state, 
another qubit is measured for the postselection of the quantum circuit to be simulated,
and yet another qubit is measured to read out the output of the simulated circuit.
This means that the DQC$1_3$ model with postselection also has the computational power equivalent to ${\postBQP = \PP}$.
By an argument similar to that in Ref.~\cite{BreJozShe11RSPA}, 
it follows that $\PP$ is in $\postBPP$ (the version of $\BPP$ with postselection),
if the DQC1$_3$ model is weakly simulatable.
Together with Toda's theorem~\cite{Tod91SIComp},
this implies the collapse of the polynomial-time hierarchy to the third level.

\paragraph{Main results.}
One obvious drawback of the existing argument above
is an inevitable postselection measurement inherent to the definition of $\postBQP$.
This becomes a quite essential obstacle when trying to extend this argument
to the DQC1 model,
where only one qubit is allowed to be measured.
To deal with the DQC1 model, this paper takes a different approach
by considering the complexity class~$\NQP$ introduced in Ref.~\cite{AdlDeMHua97SIComp}
(or $\SBQP$ introduced in Ref.~\cite{Kup09arXiv}).

Let $\NQoneP$ and $\SBQoneP$ be the versions of $\NQP$ and $\SBQP$, respectively,
in which quantum computation performed must be of DQC1 type.
First, it is proved that the classes~$\NQP$ and $\SBQP$ remain unchanged
with this restriction.

\begin{theorem}
${\NQP = \NQoneP}$ and ${\SBQP = \SBQoneP}$.
\label{Theorem: NQP = NQ[1]P and SBQP = SBQ[1]P}
\end{theorem}

The proof devises a way of simulating a given quantum circuit~$Q$ acting over $n$~qubits
in the DQC1 model with ${n+1}$~qubits.
More concretely,
this paper presents a way of letting the single clean qubit of the DQC1 model
play two roles simultaneously:
Now the clean qubit not only serves as the flag qubit
that indicates whether the remaining $n$~qubits form $\ket{0}^{\tensor n}$,
but also serves as the output qubit of the simulation of $Q$.
The simulation of $Q$ proceeds as follows:
\begin{step}
\item
  Start from the initial state ${\ketbra{0} \tensor \big(\frac{I}{2}\bigr)^{\tensor n}}$,
and flip the first qubit (the clean qubit)
if the last $n$ qubits (the completely mixed states) are in state $\ket{0}^{\otimes n}$.
\item
  Conditioned on the first qubit being $\ket{1}$,
  apply the quantum circuit~$Q$ to the last $n$ qubits.
\item
  If the state after the simulation corresponds to the accepting state in the simulated computation,
  flip the phase of the first qubit.
\item
  Perform the inverse of the unitary operation of Step~2,
  and further perform the inverse of Step~1.
\item
  Measure the first qubit in the computational basis,
  and accept iff this results in~$1$.
\end{step}
The point is that, when $Q$ realizes an NQP-type computation,
this DQC1-type computation is still a computation of NQP-type:
the acceptance probability of this DQC1-type computation is nonzero
if and only if the input is a yes-instance.
The $\SBQP$ case is proved similarly.

The main consequence of Theorem~\ref{Theorem: NQP = NQ[1]P and SBQP = SBQ[1]P} is the following,
which answers the question posed by Knill and Laflamme 
under a very plausible complexity assumption.

\begin{theorem}\label{main_thm_informal}
The DQC1 model is not weakly simulatable, unless ${\PH = \AM}$.
\end{theorem}

Towards the contraposition of the statement,
suppose that the DQC1 model were weakly simulatable.
This would in particular imply that
the above NQP-type computation in the DQC1 model is weakly simulatable.
Hence, the weak-simulatability assumption would result in a classical computation
whose acceptance probability is nonzero if and only if the input is a yes-instance,
which implies the inclusion~${\NQP \subseteq \NP}$.
Since ${\NQP = \coCequalP}$~\cite{FenGreHomPru99RSPA}, 
this inclusion is sufficient to show the collapse of the polynomial-time hierarchy to $\AM$ (and thus, to the second level),
by combining known properties in classical complexity theory. 
The same consequence can be derived by using $\SBQP$ instead of $\NQP$:
now the inclusion~${\SBQP \subseteq \SBP}$ follows from the weak-simulatability assumption,
for the class~$\SBP$ introduced in Ref.~\cite{BohGlaMei06JCSS},
and the fact~${\SBP \subseteq \AM}$ therein is used. 

\paragraph{Further results.}
The above argument based on $\NQP$ (or $\SBQP$) can replace
the existing argument based on $\postBQP$,
which was developed in Ref.~\cite{BreJozShe11RSPA}, and has appeared frequently in the literature~\cite{AarArk13ToC,Bro14arXiv,JozVan14QIC,MorFujFit14PRL,TakTanYamTan14arXiv,TakYamTan14QIC}.
In particular, it can be used even when discussing the unsimulatability of other sub-universal computing models
such as the IQP model and the Boson sampling. 
This also weakens
the complexity assumption necessary to prove the classical unsimulatability
of each of such models
(the collapse of the polynomial-time hierarchy is to the third level when using $\postBQP$,
which is now to the second level using $\NQP$ or $\SBQP$).

Finally, this paper also investigates the classical simulatability of the DQC1 and DQC1$_m$ models when additional restrictions are imposed,
as has been done for other sub-universal quantum models~\cite{BreJozShe11RSPA,Bro14arXiv,NiVan13QIC,TakTanYamTan14arXiv}. 
First, recall that
any constant-depth quantum circuit 
with a single output qubit is \emph{strongly simulatable}~\cite{FenGreHomZha05FCT},
that is, its output probability distribution is computable classically in polynomial time. 
This in particular implies that,
if the circuit used in the DQC1 model is restricted to a constant-depth one,
such a computation is strongly simulatable.
This paper extends this fact to the DQC1$_m$ model:
it is strongly simulatable if the circuit used in the DQC1$_m$ model is of constant depth.
In the case where the circuit used is of depth logarithmic,
the DQC1$_m$ model is shown \emph{not} to be weakly simulatable 
unless the polynomial-time hierarchy collapses to the second level.
This shows a clear distinction on the (weak) simulatability
between logarithmic-depth and constant-depth circuits in the DQC1$_m$ model.
Finally, it is proved that the IQP DQC1$_m$ model is strongly simulatable
for every polynomially-bounded function~$m$,
where the IQP DQC1$_m$ model is a restricted DQC1$_m$ model
such that quantum circuits used must consist of only IQP-type gates. 
As was shown in Ref.~\cite{BreJozShe11RSPA},
IQP circuits with polynomially many output qubits are not even weakly simulatable
unless the polynomial-time hierarchy collapses. 
Hence, our result suggests that the simulatability of IQP circuits changes drastically,
if the clean initial state~$\ket{0}^{\tensor n}$
is replaced by the one-clean-qubit initial state~${\ketbra{0} \tensor \bigl(\frac{I}{2}\bigr)^{\tensor (n-1)}}$.

\section{Preliminaries}\label{section_prelim}

\paragraph{Circuits and gate sets.}
A family ${\{C_w\}_{w \in \set{0,1}^\ast}}$ of randomized or quantum circuits 
is called {\em polynomial-time uniformly generated} 
if there is a deterministic algorithm that given input $w$, 
outputs a classical description of $C_w$ in time polynomial in $\abs{w}$ 
(and thus, the size of $C_w$ is also polynomially bounded in $\abs{w}$).

For quantum circuits, this paper considers only unitary quantum circuits
implemented with the gate set consisting of the Hadamard gate (the $H$ gate), 
the CNOT gate, and the gate corresponding to the unitary operator~%
${
T
=
\bigl(
\begin{smallmatrix}
  1 & 0\\
  0 & e^{i\pi/4}
\end{smallmatrix}
\bigr)
}$.
Note that all the other gates used in this paper,
such as the NOT gate (the $X$ gate), the Toffoli gate, the $Z$ gates, 
and the controlled-$Z$ gate,
are implemented exactly in this gate set without any ancilla qubits.
The only exception is the generalized Toffoli gate
(i.e., the $k$-controlled-NOT gate~$\land_k(X)$),
which is also exactly implementable with this gate set,
but using some ancilla qubits,
according to the constructions in Ref.~\cite{Bar+95PRA}.
In the construction in Lemma~7.2 of Ref.~\cite{Bar+95PRA},
the number of necessary ancilla qubits
grows linearly with respect to the number of control qubits,
and the construction in Corollary~7.4 of Ref.~\cite{Bar+95PRA}
uses only two ancilla qubits when ${k \geq 5}$,
but an important property is that
no initializations are required for all these ancilla qubits
(and thus, all of them can actually be re-used 
when applying other generalized Toffoli gates).
In particular, even completely mixed states may be used 
for these ancilla qubits, and hence,
generalized Toffoli gates may be assumed available freely in the DQC1 model
with the above-mentioned gate set.
See Lemma~7.2 and Corollary~7.4 in Ref.~\cite{Bar+95PRA} for details.
For simplicity, in what follows,
we identify the quantum circuit~$Q$ with the unitary operator it induces.

\paragraph{DQC1 model.}
A quantum computation of DQC1 type
is a computation performed by a unitary quantum circuit~$Q$. 
It is assumed that one of the qubits to which the circuit~$Q$ is applied
is designated as the output qubit.
The DQC1-type computation specified by the circuit~$Q$ proceeds as follows.
Let $n$ denote the number of qubits $Q$ acts over.
The initial state of the computation is the $n$-qubit state
${
\rho_\init = \ketbra{0} \tensor (I/2)^{\tensor (n - 1)}
}$.
The circuit~$Q$ is then applied to this initial state,
which generates the $n$-qubit state
${
\rho_\final
=
Q \rho_\init \conjugate{Q}
}$.
Now the designated output qubit is measured in the computational basis,
where the outcome~$1$ is interpreted as ``accept''
and the outcome~$0$ is interpreted as ``reject''. 
A quantum computation of DQC1$_m$ type is defined similarly to the DQC1 case,
except that $m$~qubits are designated as the output qubits to be measured.

\paragraph{Classical simulatability.}
Following conventions,
this paper uses the following definitions of simulatability.
Consider any family~${\{ Q_w \}_{w \in \{0,1\}^\ast}}$ of quantum circuits, 
and for each circuit~$Q_w$,
suppose that $m$~output qubits are measured in the computational basis
after the application of $Q_w$ to a certain prescribed initial state
(which will be clear from the context).
Let $\function{P_w}{\{0,1\}^m}{[0,1]}$ be the probability distribution
derived from the output of $Q_w$
(i.e., ${P_w(x_1, \ldots, x_m)}$ is the probability of obtaining the measurement result~${(x_1, \ldots, x_m)}$ in ${\{0,1\}^m}$ when $Q_w$ is applied to the prescribed initial state).

The family ${\{Q_w\}_{w \in \{0,1\}^\ast}}$ is 
\emph{weakly simulatable with multiplicative error~${c \geq 1}$}
if there exists a family~${\{P'_w\}}$ of probability distributions
that can be sampled classically in polynomial time
such that, for any $w$ in ${\{0,1\}^\ast}$ and ${(x_1, \ldots, x_m)}$ in ${\{0,1\}^m}$, 
\begin{equation}
\frac{1}{c}P_w(x_1, \ldots, x_m) \leq P'_w(x_1, \ldots, x_m) \leq c P_w(x_1, \ldots, x_m).
\label{approximation}
\end{equation}

The family ${\{Q_w\}_{w \in \{0,1\}^\ast}}$ is \emph{weakly simulatable 
with exponentially small additive error}
if, for any polynomially bounded function~$q$,
there exists a family~${\{P'_w\}}$ of probability distributions
that can be sampled classically in polynomial time
such that, for any $w$ in ${\{0,1\}^\ast}$ and ${(x_1, \ldots, x_m)}$ in ${\{0,1\}^m}$,
\[
\abs{P_w(x_1, \ldots, x_m) - P'_w(x_1, \ldots, x_m)} \leq 2^{-q(\abs{w})}.
\]

\begin{remark}
The notion of weak simulatablity with multiplicative error
was first defined in Ref.~\cite{TerDiV04QIC} in a slightly different form.
The definition taken in this paper is found 
in Refs.~\cite{BreJozShe11RSPA,MorFujFit14PRL}, for instance.
The version in Ref.~\cite{TerDiV04QIC} uses
${\abs{P_w(x_1, \ldots, x_m) - P_w'(x_1, \ldots, x_m)} \leq \varepsilon P_w(x_1, \ldots, x_m)}$
instead of the bounds~(\ref{approximation}),
and these two versions are essentially equivalent.
The results in this paper hold for any $\varepsilon$ in ${[0,1)}$
when using the version of Ref.~\cite{TerDiV04QIC}.
The notion of weak simulatability with exponentially small additive error
was introduced in Ref.~\cite{TakYamTan14QIC},
and was used also in Ref.~\cite{TakTanYamTan14arXiv}.
\end{remark}

\begin{remark}
As many existing studies adopt the weak simulatablity with multiplicative error
when discussing the classical simulatability of quantum models
(Refs.~\cite{TerDiV04QIC,BreJozShe11RSPA,AarArk13ToC,NiVan13QIC,Bro14arXiv,JozVan14QIC,MorFujFit14PRL}, for instance),
it would be sufficiently reasonable to use this notion also in the case of the DQC1 model,
which in particular makes it possible to discuss the power of the DQC1 model
along the line of these existing studies.
On the contrary, as discussed in Refs.~\cite{AarArk13ToC,BreJozShe11RSPA}, 
the notion of weak simulatability with \emph{polynomially} small additive error
is perhaps much more desirable for experimentally verifying the superiority of a quantum computation model.
Proving or disproving classical simulatability under this notion
is one of the most important open problems in most of sub-universal quantum models
including the DQC1 model.
\end{remark}

\paragraph{Complexity classes.}
For a quantum circuit~$Q$,
let ${p_\acc(Q)}$ (resp., ${p^\DQCone_\acc(Q)}$)
denote the acceptance probability of $Q$
(resp., the acceptance probability of $Q$ in the DQC1 model),
i.e., the probability that
the measurement on the designated output qubit of $Q$
in the computational basis results in~$1$
after $Q$ is performed with the initial state~$\ket{0}^{\tensor n}$
(resp., ${\ketbra{0} \tensor \bigl( \frac{I}{2} \bigr)^{\tensor (n-1)}}$), 
where $n$ is the number of qubits $Q$ acts over.
For a randomized circuit~$C$,
its acceptance probability ${p_\acc(C)}$ is defined as
the probability that the designated output bit of $C$ is~$1$
after $C$ is performed with input~$0^n$,
where $n$ is the number of bits $C$ acts over.

A language $L$ is in $\NQP$ 
iff there exists a polynomial-time uniformly generated family 
${\{ C_w \}_{w \in \{0,1\}^\ast}}$ of quantum circuits such that 
for every ${w \in \{0,1\}^\ast}$,
(i) if $w$ is in $L$, ${p_\acc(C_w) > 0}$,
and (ii) if $w$ is not in $L$, ${p_\acc(C_w) = 0}$.

A language $L$ is in $\SBP$ (resp., $\SBQP$) 
iff there exist a polynomially bounded function~$q$ 
and a polynomial-time uniformly generated family~${\{ C_w \}_{w \in \{0,1\}^\ast}}$
of randomized (resp., quantum) circuits such that, 
for every ${w \in \{0,1\}^\ast}$,
(i) if $w$ is in $L$, ${p_\acc(C_w) \geq 2^{-q(\abs{w})}}$,
and 
(ii) if $w$ is not in $L$, ${p_\acc(C_w) \leq 2^{-q(\abs{w})-1}}$.

Note that the classes~$\SBP$ and $\SBQP$
remain unchanged even when the two thresholds $2^{-q(\abs{w})}$ and $2^{-q(\abs{w})-1}$ 
are replaced by ${c 2^{-q(\abs{w})}}$ and ${c' 2^{-q(\abs{w})}}$, respectively,
for any two constants~$c$ and $c'$ satisfying ${0 < c' < c \leq 1}$.

The complexity classes~$\NQoneP$ and $\SBQoneP$ are defined by
replacing every ${p_\acc(C_w)}$ with ${p^\DQCone_\acc(C_w)}$
in the definitions of $\NQP$ and $\SBQP$, respectively.

\section{Main Results}

To prove Theorem~\ref{Theorem: NQP = NQ[1]P and SBQP = SBQ[1]P},
we start with analyzing the DQC1-type computation presented in Section~\ref{Section: introduction}.

Consider any polynomial-time uniformly generated family~${\{Q_w\}_{w \in \set{0,1}^\ast}}$ of quantum circuits.
Fix $w$ in ${\set{0,1}^\ast}$,
and let $n$ be the number of qubits $Q_w$ acts over.
Without loss of generality, it is assumed that
the first qubit is the designated output qubit of $Q_w$.
Let $p_w$ be the probability $p_\acc(Q_w)$ of obtaining the classical outcome~$1$
when measuring the output qubit of $Q_w$ in the computational basis, 
i.e., 
\[ 
p_w=\norm{\Pi_1Q_w\ket{0}^{\otimes n}}^2,
\]
where $\Pi_1$ is the projection~${\ketbra{1} \tensor I^{\tensor (n-1)}}$.

\begin{figure}[t!]
\begin{center}
\hspace{-4mm}
\includegraphics[scale=0.6]{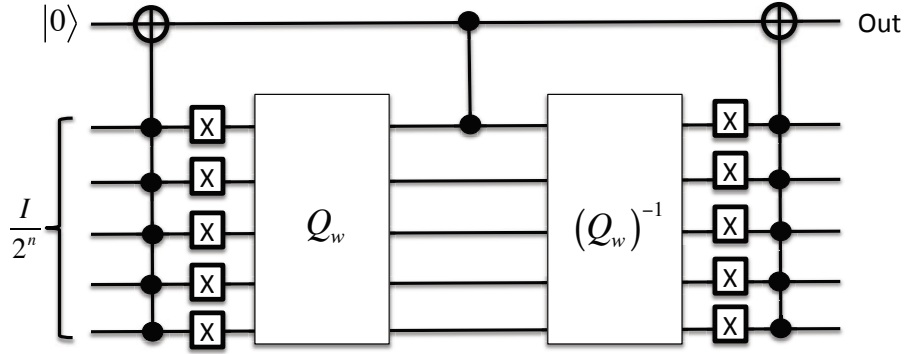}
\caption{Circuit $D_w$. Note that the first and the last layers are generalized Toffoli gates, 
where black circles are control parts and $\oplus$ is the target part, and the gate represented by two black circles connected by a vertical line is a controlled-$Z$ gate.}
\label{Fig:DQC1_1}
\end{center}
\end{figure}

Consider the quantum circuit $D_w$ depicted in Fig.~\ref{Fig:DQC1_1},
which acts over ${n + 1}$ qubits.
As indicated in Fig.~\ref{Fig:DQC1_1},
the first qubit is the output qubit of $D_w$.
We analyze the DQC1-type computation induced by $D_w$.

\begin{lemma}\label{DQC1_1circuit}
\label{lm:main}
Let $\tilde{p}_w$ be 
the probability that the measurement on the output qubit of $D_w$
in the computational basis results in~$1$
after $D_w$ is performed with the initial state~${\ketbra{0} \tensor \bigl( \frac{I}{2} \bigr)^{\tensor n}}$.
Then,
$\tilde{p}_w=\frac{4}{2^n}p_w(1-p_w)$.
\end{lemma}

\begin{proof}
First, consider the case where the first generalized Toffoli gate is activated, 
meaning that the content of the first qubit is flipped from 0 to 1
and all the $n$ input qubits of $Q_w$ are set to the state $\ket{0}^{\tensor n}$.
This occurs with probability~$\frac{1}{2^n}$.
In this case, the controlled-$Z$ gate flips the phase 
if and only if the content of the output qubit of $Q_w$ is 1, 
that is, this gate performs the operation ${I^{\tensor n} - 2 \Pi_1}$ over the last $n$~qubits. 
Hence, the state of the last $n$~qubits
immediately after applying $\conjugate{Q_w}$ is given by
${
\conjugate{Q_w} \bigl( I^{\tensor n} - 2 \Pi_1 \bigr) Q_w \ket{0}^{\tensor n}
}$,
and thus, the last generalized Toffoli gate is activated with probability
\[
\bigabs{
  \bra{0}^{\tensor n} \conjugate{Q_w} \bigl( I^{\tensor n} - 2 \Pi_1 \bigr) Q_w \ket{0}^{\tensor n}
}^2
=
\bigabs{
  1 - 2 \norm{\Pi_1 Q_w \ket{0}^{\tensor n}}^2
}
=
(1 - 2p_w)^2.
\]
The probability that the last generalized Toffoli gate is \emph{not} activated
is hence $1-(1-2p_w)^2=4p_w(1-p_w)$. With this probability, the outcome of $D_w$ is 1
(on the condition that the first generalized Toffoli is activated).

Now consider the case where the first generalized Toffoli gate is \emph{not} activated. 
In this case, the first qubit is still $\ket{0}$ after applying the generalized Toffoli gate.
Thus, the controlled-$Z$ gate is not activated, implying that $Q_w$ and $\conjugate{Q_w}$
are cancelled.
Therefore, the last generalized Toffoli gate is not activated, either,
and the outcome of $D_w$ is always $0$.

The total probability of obtaining the outcome 1 is thus $\frac{4}{2^n}p_w(1-p_w)$. 
\end{proof}

From Lemma~\ref{lm:main},
Theorem~\ref{Theorem: NQP = NQ[1]P and SBQP = SBQ[1]P} is easily proved as follows.

\begin{proof}[Proof of Theorem~\ref{Theorem: NQP = NQ[1]P and SBQP = SBQ[1]P}]
Suppose that a language $L$ is in $\NQP$,
and let ${\{Q_w\}}$ 
be a polynomial-time uniformly generated family of quantum circuits 
that witnesses this fact.
By the definition of $\NQP$,
for every ${w \in \set{0,1}^\ast}$,
the acceptance probability $p_w$ of the circuit~$Q_w$
is positive if and only if $w$ is in $L$.
Without loss of generality, one can assume that ${p_w < 1}$ for every $w$.

From each $Q_w$, we construct a quantum circuit~$D_w$ as defined in Fig.~\ref{Fig:DQC1_1},
which provides the polynomial-time uniformly generated family~${\{D_w\}}$ of quantum circuits. 
From Lemma~\ref{lm:main},
this ${\{D_w\}}$ ensures that $L$ is in $\NQoneP$, 
as the probability~$\tilde{p}_w$ associated with $D_w$ is nonzero
if and only if ${0 < p_w < 1}$.
The other containment~${\NQoneP \subseteq \NQP}$ is trivial.
That ${\SBQP = \SBQoneP}$ can be proved similarly\footnote{Notice that without loss of generality we can assume that $n$ depends on $|w|$ only, 
and thus the proof can be done based on the definition of $\SBQP$, instead of $\NQP$.}.
\end{proof}

Next, we give a formal statement of Theorem~\ref{main_thm_informal}.

\begin{theorem}\label{DQC1_1}
If any polynomial-time uniformly generated family of quantum circuits,
when used in the DQC1-type computations, 
is weakly simulatable with multiplicative error $c\geq 1$ 
or exponentially small additive error, $\PH=\AM$.
\end{theorem}

Theorem \ref{DQC1_1} follows directly from Lemmas~\ref{lm:simulation}, 
\ref{lem:NQP=NP}, and \ref{lem:NQP in SBP} below.

\begin{lemma}
\label{lm:simulation}
If any polynomial-time uniformly generated family of quantum circuits,
when used in the DQC1-type computations, 
is weakly simulatable with multiplicative error $c\geq 1$ 
(resp., exponentially small additive error), then ${\NQP=\NP}$ (resp., ${\NQP \subseteq \SBP}$). 
\end{lemma}
\begin{proof}
Fix any language~$L$ in $\NQP$. 
By Theorem~\ref{Theorem: NQP = NQ[1]P and SBQP = SBQ[1]P},
 there is a polynomial-time uniformly generated family~${\{D_w\}}$ of quantum circuits
such that, for every ${w \in \{0, 1 \}^\ast}$,
the acceptance probability $\tilde{p}_w$
in the DQC1-type computation induced by $D_w$
is nonzero if and only if $w$ is in $L$.
From the assumption of this lemma,
there exists a polynomial-time uniformly generated family~${\{C_w\}}$
of randomized circuits such that weakly simulates ${\{D_w\}}$
with multiplicative error~$c\geq 1$. 
By the definition, 
the probability that the circuit~$C_w$ outputs 1 is nonzero
if and only if $\tilde{p}_w$ is nonzero,
which happens only when $w$ is in $L$.
This implies that $L$ is in $\NP$.

In the case where ${\{C_w\}}$ weakly simulates ${\{D_w\}}$ 
with exponentially small additive error,
we use the fact that
$\tilde{p}_w$ with each $w$ in $L$ may be assumed to be bounded from below by $2^{-q(\abs{w})}$
for a certain polynomially bounded function~$q$.
This implies that the probability that $C_w$ outputs 1 when $w$ is in $L$
must differ from the probability when $w$ is not in $L$
by at least a constant multiplicative factor,
if we set a sufficiently small additive error, say, $2^{-q(\abs{w})}/1000$,
which shows that $L$ is in $\SBP$.
\end{proof}

Let $\calC$ be any class of languages. A language $L$ is in $\BP\cdot\calC$ 
if there exist a language $A\in\calC$ and a polynomially bounded function~$r$ such that for every $w\in\{0,1\}^\ast$,
\[
|\Set{z\in\{0,1\}^{r(\abs{w})} }{w\in L~~\mbox{iff}~~\langle w,z\rangle\in A}| \geq 
\frac{2}{3}\cdot 2^{r(\abs{w})}.
\] 
It is easy to see that $\AM=\BP\cdot\NP$. In Ref.~\cite{TodOgi92SIComp}, the following variant of the BP operator was defined: $L$ is in $\widehat{\BP}\cdot \calC$ 
if for every polynomially bounded function~$q$, there exist a language $A\in\calC$ and a polynomially bounded function~$r$
such that for every ${w\in\{0,1\}^\ast}$,
\[
|\Set{z\in\{0,1\}^{r(\abs{w})} }{w\in L~~\mbox{iff}~~\langle w,z\rangle\in A}|
\geq (1-2^{-q(|w|)})2^{r(\abs{w})}.
\] 
By the standard amplification of $\AM$, one can see that $\AM=\BP\cdot\NP=\widehat{\BP}\cdot\NP$.
 
\begin{lemma}\label{lem:NQP=NP}
If $\NQP=\NP$, then $\PH=\AM$.
\end{lemma}

\begin{proof}
This follows from the following sequence of containments: 
\[
\PH \subseteq \BP \cdot \coCequalP = \BP \cdot \NQP = \BP \cdot \NP = \AM,
\]
where the first containment is by Corollary~2.5 in Ref.~\cite{TodOgi92SIComp} or Corollary~5.2 in Ref.~\cite{Tar93TCS},
and the next two equalities follow from the fact~${\NQP = \coCequalP}$~\cite{FenGreHomPru99RSPA}
and the assumption ${\NQP = \NP}$ of this lemma, respectively.
\end{proof}

\begin{lemma}\label{lem:NQP in SBP}
If $\NQP\subseteq\SBP$, then $\PH=\AM$.
\end{lemma}

\begin{proof}
The claim follows from the following sequence of containments:
\[
\PH
\subseteq
\widehat{\BP} \cdot \coCequalP
=
\widehat{\BP} \cdot \NQP
\subseteq
\widehat{\BP} \cdot \SBP
\subseteq
\widehat{\BP} \cdot \widehat{\BP} \cdot \NP
=
\widehat{\BP} \cdot \NP
=
\AM,
\]
where the first inclusion is by Corollary 2.5 in Ref.~\cite{TodOgi92SIComp},
and we have used the fact~${\NQP = \coCequalP}$~\cite{FenGreHomPru99RSPA},
the assumption~${\NQP \subseteq \SBP}$ of this lemma,
and the facts~${\SBP \subseteq \AM}$~\cite{BohGlaMei06JCSS}
and ${\widehat{\BP} \cdot \NP = \AM}$,
as well as Lemma~2.8 in Ref.~\cite{TodOgi92SIComp} on the removability of a duplicate $\widehat{\BP}$ operator.
\end{proof}

\begin{remark}
The collapse of $\PH$ to $\AM$ follows from the assumption~${\SBQP = \SBP}$, too,
as ${\NQP \subseteq \SBQP}$ holds (note that the acceptance probability 
of NQP-type computation on any yes-instance is at least $1/2^q$ 
for some polynomial $q$).
\end{remark}

The following theorem is obtained by Lemmas \ref{lem:NQP=NP} 
and \ref{lem:NQP in SBP}, and a slight modification of the proof of Lemma \ref{lm:simulation},  
in which Theorem~\ref{Theorem: NQP = NQ[1]P and SBQP = SBQ[1]P}
 is replaced by the assumption of the theorem.

\begin{theorem}\label{general_collapse}
Assume that given a polynomial-time uniformly generated family $\{Q_w\}$ 
of quantum circuits, 
there is a polynomial-time uniformly generated family $\{D_w\}$ of quantum circuits 
conforming to a computing model ${\cal M}$ whose acceptance probability 
is positive (resp.~at least $\frac{p_\acc(Q_w)}{2^{r(|w|)}}$ for some polynomial $r$) 
if $p_\acc(Q_w)>0$, and zero if $p_\acc(Q_w)=0$. 
If ${\cal M}$ is weakly simulatable with multiplicative error $c\geq 1$ 
(resp.~exponentially small additive error), then $\PH=\AM$.
\end{theorem}

Theorem \ref{general_collapse} implies that the collapse 
of the polynomial-time hierarchy in the previous results 
 on the IQP model~\cite{BreJozShe11RSPA}, the Boson sampling model~\cite{AarArk13ToC}, etc.~can be also improved from the third level to the second level. 
In fact, for all the previous results using the argument based on $\post\BQP$ 
\cite{AarArk13ToC,BreJozShe11RSPA,Bro14arXiv,JozVan14QIC,MorFujFit14PRL,TakTanYamTan14arXiv,TakYamTan14QIC},  
we can apply Theorem \ref{general_collapse} to them,  
since the corresponding circuit $D_w$ is a postselected circuit,
and the outcome of $D_w$ is considered to be accept 
if and only if the output qubit and the postselection qubit are both $1$. 
Moreover, our argument clarifies that Theorem~\ref{general_collapse} holds 
for any multiplicative constant $c\geq 1$, while $c$ is restricted to $1\leq c<\sqrt{2}$ 
in Refs.~\cite{BreJozShe11RSPA,MorFujFit14PRL}. 

\section{DQC1$_m$ model with additional restrictions}

\subsection{Depth-restricted quantum circuits}

First, we consider the case where the depth of quantum circuits, 
when used in the DQC1$_m$-type computations,   
is logarithmic in the length of the input. 
For the case where the number of output qubits is not restricted, 
it is shown that the classical simulation is hard 
unless the polynomial-time hierarchy collapses 
(note that for $m=1$, i.e., the DQC1 type, we do not know whether a similar result holds).

\begin{theorem}\label{DQC1_k-log_depth}
If any polynomial-time uniformly generated family 
of  logarithmic-depth quantum circuits, 
when used in the DQC1$_m$-type computations,  
is weakly simulatable with multiplicative error $c\geq 1$ 
(or exponentially small additive error)
for sufficiently large $m$, then $\PH=\AM$.
\end{theorem}

\begin{proof}
We only give the proof for the case of multiplicative error 
as its modification to the case of exponentially small additive error is similar to the proof of Theorem~\ref{DQC1_1}. 

Let us assume that a language $L$ is in $\NQP$. This means that
there exists a polynomial-time uniformly generated family $\{Q_w\}$ 
of quantum circuits such that 
the output probability distribution $P_{Q_w}$ satisfies: 
if $w\in L$ then $P_{Q_w}(o=1)>0$, and if $w\notin L$ then $P_{Q_w}(o=1)=0$, 
where $o$ is the single-bit output of the quantum circuit $Q_w$. 
By using the gate teleportation technique~\cite[Lemma 1]{FenGreHomZha05FCT}, 
we can construct from $Q_w$ a constant-depth quantum circuit $Q'_w$
which has $(r+1)$ output qubits $(o,p_1, \ldots ,p_r)$, where $r$ is polynomially 
bounded in $|w|$, 
such that its output probability distribution $P_{Q'_w}$ satisfies
\begin{eqnarray*}
P_{Q'_w}(o=1|p_1=\cdots=p_r=1)=P_{Q_w}(o=1).
\end{eqnarray*}
From $Q'_w$, we further construct a logarithmic-depth quantum circuit
$V_w$ as follows: 
\begin{enumerate} 
\item Prepare the qubits used to simulate $Q'_w$ and an extra qubit (called qubit $o'$) 
to $\ket{0}$.
\item Apply $Q'_w$ on the qubits except $o'$.
\item Prepare more extra qubits to $\ket{0}$ for using Toffoli gates as (binary) AND gates, 
and take the $r$-ary AND of all the qubits $(p_1,\ldots,p_r)$ 
(let $p'$ be one of the extra qubits used for the result of the $r$-ary AND). 
\item Set $o'=0$ if and only if $o=1$ and $p'=1$.
\end{enumerate}
As step 3 is implemented in $O(\log |w|)$ depth, 
it can be seen that the depth of $V_w$ is at most logarithmic in $|w|$. 
Also, it is easy to verify that the output distribution $P_{V_w}$ satisfies 
\begin{eqnarray*}
P_{V_w}(o'=0)=\frac{P_{Q_w}(o=1)}{2^r}.
\end{eqnarray*}

Let us assume that any logarithmic depth quantum circuit, 
when used in the DQC1$_m$-type computations, 
 can be classically efficiently sampled with multiplicative error $c\geq 1$.  
Assume that $V_w$ acts on $l+1$ qubits, where $l$ is polynomially bounded in $|w|$. 
Now consider the DQC1$_{l+1}$-type computation specified by the circuit $V_w^\dagger$. 
Let $P'_w$ be the output probability distribution of the classical simulator 
of this computation, which has $l+1$ output ports $(o_1, \ldots, o_{l+1})$.

Then, if $w\in L$,
\begin{eqnarray*}
P'_w(o_1=\cdots=o_{l+1}=0)
&\ge&
\frac{1}{c}\mbox{Tr}
\Big[(|0\rangle\langle 0|)^{\otimes (l+1)}
\times V_w^\dagger
\Big(|0\rangle\langle0|\otimes\frac{I^{\otimes l}}{2^l}
\Big)
V_w
\Big]\\
&=&
\frac{1}{c2^l}\mbox{Tr}\Big[
\Big(|0\rangle\langle0|\otimes I^{\otimes l}\Big)\times
V_w(|0\rangle\langle 0|)^{\otimes (l+1)}V_w^\dagger
\Big]\\
&=&
\frac{P_{Q_w}(o=1)}{c2^{l+r}}\\
&>&0.
\end{eqnarray*}
On the other hand,
if $w\notin L$,
\begin{eqnarray*}
P'_w(o_1=\cdots=o_{l+1}=0)
&\le&
c\mbox{Tr}\Big[(|0\rangle\langle 0|)^{\otimes (l+1)}\times V_w^\dagger
\Big(|0\rangle\langle0|\otimes\frac{I^{\otimes l}}{2^l}
\Big)V_w\Big]\\
&=&
\frac{c}{2^l}
\mbox{Tr}\Big[
\Big(|0\rangle\langle0|\otimes I^{\otimes l}\Big)\times
V_w(|0\rangle\langle 0|)^{\otimes (l+1)}V_w^\dagger
\Big]\\
&=&
\frac{c P_{Q_w}(o=1)}{2^{l+r}}\\
&=&0.
\end{eqnarray*}
This means that $L$ is in $\NP$. Therefore, $\NQP\subseteq\NP$, 
which leads to $\PH=\AM$.
\end{proof}

On the contrary, if we further restrict the depth to a constant,
then the DQC1$_m$ model is classically simulatable in the strong sense.
This contrasts with the fact that the IQP model is hard to simulate classically 
even if the depth is restricted to a constant~\cite{BreJozShe11RSPA}
unless the polynomial-time hierarchy collapses.

\begin{theorem}\label{DQC1_k-constant_depth}
Any polynomial-time uniformly generated family of
  constant-depth (or even doubly logarithmic depth)
  quantum circuits, when used in the DQC1$_m$-type computations, 
is strongly simulatable for every $m$.
In other words, any marginal distribution of the output of the circuit 
can be calculated with a classical polynomial-time computer.
\end{theorem}

\begin{proof}
Let $Q$ be an $(l+1)$-qubit constant-depth quantum circuit, 
when used in the DQC1$_{l+1}$-type computations, 
where $l$ is a polynomial in the input length (here we omit the subscript of the circuit name 
which represents the input). 
For any $z=(z_1,\ldots, z_{l+1})\in\{0,1\}^{l+1}$, 
let $P(z)$ be the probability that the output of $Q$ is $z$. Then, 
\begin{eqnarray*}
P(z)&=&
\mbox{Tr}
\Big[|z\rangle\langle z|\times Q
\Big(|0\rangle\langle0|\otimes \frac{I^{\otimes l}}{2^l} 
\Big)Q^\dagger\Big]\\
&=&\mbox{Tr}
\Big[
\Big(\bigotimes_{j=1}^{l+1}X^{z_j}\Big)
(|0\rangle\langle 0|)^{\otimes (l+1)}
\Big(\bigotimes_{j=1}^{l+1}X^{z_j}\Big)\times
Q\Big(|0\rangle\langle0|\otimes\frac{I^{\otimes l}}{2^l}\Big)Q^\dagger\Big]\\
&=&\frac{1}{2^l}\mbox{Tr}
\Big[
\Big(|0\rangle\langle0|\otimes I^{\otimes l}\Big)\times
Q^\dagger
\Big(\bigotimes_{j=1}^{l+1}X^{z_j}\Big)
(|0\rangle\langle 0|)^{\otimes (l+1)}
\Big(\bigotimes_{j=1}^{l+1}X^{z_j}\Big)
Q
\Big],
\end{eqnarray*}
which can be exactly calculated for any $z$, 
since the numerator is the single-qubit output probability 
of the constant-depth circuit $Q^\dagger\left(\bigotimes_{j=1}^{l+1}X^{z_j}\right)$ 
with the input $|0\rangle^{\otimes (l+1)}$.
(Note that the single-qubit output probability distribution 
of any constant-depth quantum circuit can be exactly calculated with
a polynomial-time classical computer~\cite{FenGreHomZha05FCT}, 
since the single output qubit is entangled only a constant number of input qubits, 
and the quantum computing of a constant number of qubits can be easily simulated 
with a classical computer.)

Any marginal of the probability distribution $\{P(z)\}$ can also be exactly calculated, since
in this case we have only to consider a constant-depth quantum circuit
where some of clean input qubits are replaced with completely-mixed states.

Finally, it is easily seen that the above proof works even if the depth of $Q$ 
is at most $d=\log_2 (c\cdot\log_2 n)$ for the input length $n$ and 
a positive constant $c$, since in this case the single output qubit is 
affected by at most $2^d=c\cdot \log_2 n$ input qubits, and thus 
we only need to consider a subspace of $2^{2^d}=n^c$ dimension.
\end{proof}

\subsection{IQP DQC1$_m$ model}

Finally, we consider the ``intersection'' of the IQP model and the DQC1 model.
An IQP DQC1$_m$ circuit is defined in the following way: 
For any polynomially-bounded $l$ and any $m\le l+1$,
\begin{enumerate}
\item the input state is $|0\rangle\langle 0| \otimes (\frac{I}{2})^{\otimes l}$;
\item apply $H^{\otimes (l+1)}$ to all the qubits;
\item apply a polynomial number of the controlled-$Z$ gates and the $e^{i \theta Z}$ gates\footnote{Here, $\theta$ is a multiple of $\pi/8$ due to our choice of the gate sets. 
However, Theorem~\ref{IQP-DQC1} holds when $\theta$ is any polynomial-time 
computable real, if arbitrarily small approximation error is allowed for the strong simulation. 
Moreover, it still holds even if the controlled-$Z$ gates are extended to the $e^{i \theta (Z_{j_1}\otimes Z_{j_2})}$ gates (where $Z_j$ denotes the $Z$ gate acting on qubit $j$).};
\item apply $H^{\otimes (l+1)}$ to all the qubits;
\item measure $m$ output qubits in the computational basis.
\end{enumerate} 

Ref.~\cite{BreJozShe11RSPA} shows that 
IQP circuits with polynomially many output qubits are not even weakly simulatable
unless the polynomial hierarchy collapses.
Here, we can show that any IQP DQC1$_m$ circuit is strongly simulatable for any~$m$. 

\begin{theorem}\label{IQP-DQC1}
Any polynomial-time uniformly generated family of  IQP DQC1$_m$ circuits 
is strongly simulatable for every $m$.
\end{theorem}

\begin{proof}
Let $Q$ be an $(l+1)$-qubit IQP DQC1$_{l+1}$ circuit. 
Since the controlled-$Z$ gate and the $e^{i\theta Z}$ gate 
commute with each other, we can write $Q$ as
\begin{eqnarray*}
Q=H^{\otimes (l+1)}\Big(\bigotimes_{j=1}^{l+1}e^{i\theta_j Z}\Big)
\Big(\prod_{(i,j)\in E}CZ_{i,j}
\Big)H^{\otimes (l+1)}
\end{eqnarray*}
without loss of generality,
where $E$ is the set of edges of a certain $(l+1)$-vertex graph $G$,  
and $CZ_{i,j}$ denotes the operator on the $(l+1)$ qubits that corresponds 
to the controlled-$Z$ gate applied to the two qubits 
on vertices $i$ and $j$. 
Furthermore, it is easy to verify that $Q$ commutes with 
$\bigotimes_{j=1}^{l+1}X^{x_j}$ for any $(x_1,\ldots, x_{l+1})\in\{0,1\}^{l+1}$.

For any $z=(z_1,\ldots, z_{l+1})\in\{0,1\}^{l+1}$, let $P(z)$ be the probability 
that the output of $Q$ is $z$. Then, 
\begin{eqnarray*}
P(z)
&=&\mbox{Tr}
\Big[|z\rangle\langle z|\times
Q\Big(|0\rangle\langle0|\otimes\frac{I^{\otimes l}}{2^l}\Big)Q^\dagger\Big]\\
&=&
\frac{1}{2^l}\sum_{
s\in\{0,1\}^l}\mbox{Tr}
\Big[|z\rangle\langle z|\times 
Q
\Big(I\otimes\bigotimes_{j=1}^l X^{s_j}\Big)
\Big(|0\rangle\langle0|\otimes(|0\rangle\langle0|)^{\otimes l}\Big)
\Big(I\otimes\bigotimes_{j=1}^l X^{s_j}\Big)
Q^\dagger\Big]\\
&=&
\frac{1}{2^l}\sum_{
s\in\{0,1\}^l}\mbox{Tr}
\Big[
\Big(I\otimes\bigotimes_{j=1}^l X^{s_j}\Big)
|z\rangle\langle z|
\Big(I\otimes\bigotimes_{j=1}^l X^{s_j}\Big)
\times
Q
\Big(
(|0\rangle\langle0|)^{\otimes (l+1)}\Big)
Q^\dagger\Big]\\
&=&
\frac{1}{2^l}\sum_{
s\in\{0,1\}^l}\mbox{Tr}
\Big[
\Big(|z_1\rangle\langle z_1|\otimes\bigotimes_{j=1}^l
|z_{j+1}\oplus s_j\rangle\langle z_{j+1}\oplus s_j|\Big)
\times
Q
\Big(
(|0\rangle\langle0|)^{\otimes (l+1)}\Big)
Q^\dagger\Big]\\
&=&
\frac{1}{2^l}\mbox{Tr}
\Big[
(|z_1\rangle\langle z_1|\otimes I^{\otimes l})\times
Q
\Big(
(|0\rangle\langle0|)^{\otimes (l+1)}\Big)
Q^\dagger\Big]\\
&=&
\frac{1}{2^l}\mbox{Tr}
\Big[
(|\phi_{z_1}\rangle\langle \phi_{z_1}|\otimes I^{\otimes l})\times
|G\rangle\langle G|
\Big]
\end{eqnarray*}
where $s_j$ is the $j$th bit of $s$, 
\begin{eqnarray*}
|G\rangle=\Big(\prod_{(i,j)\in E}CZ_{i,j}\Big)|+\rangle^{\otimes (l+1)}
\end{eqnarray*}
is a graph state (recall that $|+\rangle=\frac{1}{\sqrt{2}}(\ket{0}+\ket{1})$),
and
\begin{eqnarray*}
|\phi_{z_1}\rangle=e^{-i\theta_1 Z}H|z_1\rangle.
\end{eqnarray*}

Note that $p(z_1)\equiv \mbox{Tr}
\Big[
(|\phi_{z_1}\rangle\langle \phi_{z_1}|\otimes I^{\otimes m})\times
|G\rangle\langle G|
\Big]$ determines 
the output probability distribution of
the single-qubit measurement on the graph state. 
This can be calculated in classical polynomial time:  
If the vertex of $G$ corresponding to the first qubit
is isolated, $p(z_1)$ can be easily calculated
since a single qubit quantum computing can be trivially simulated 
with a classical computer.
If the vertex is not isolated, $p(z_1)=\frac{1}{2}$,
since
\begin{eqnarray*}
|G\rangle=\frac{1}{\sqrt{2}}\Big[
|0\rangle\otimes |G'\rangle
+|1\rangle\otimes \Big(\bigotimes_{(1,j)\in E}Z_j\Big)|G'\rangle
\Big],
\end{eqnarray*}
where $Z_j$ is the $Z$ gate applied to the qubit on vertex $j$, 
$G'$ is the $m$-vertex graph created from $G$ by removing the first vertex
and edges connected with the first vertex,
and $|G'\rangle$ and 
$\Big(\bigotimes_{(1,j)\in E}Z_j\Big)|G'\rangle$ are orthogonal with each other.
Thus, $P(z)$ is calculated in classical polynomial time.
\end{proof}

\subsection*{Acknowledgements}
Keisuke Fujii is supported by JSPS Grant-in-Aid for Research Activity Start-up 25887034. 
Hirotada Kobayashi and Harumichi Nishimura are supported by KAKENHI 24240001. 
Tomoyuki Morimae is supported by the Tenure Track System by MEXT Japan, and KAKENHI 26730003.
Harumichi Nishimura is also supported by KAKENHI 24106009 and 25330012.
Hirotada Kobayashi and Seiichiro Tani are also grateful to KAKENHI 24106009.


\providecommand{\noopsort}[1]{}

\end{document}